\documentclass[english]{lipics-v2021}
\nolinenumbers

\usepackage[utf8]{inputenc}

\usepackage{microtype,hyperref,graphicx}
   \hypersetup{%
      breaklinks,%
      colorlinks=true,%
      urlcolor=[rgb]{0.25,0.0,0.0},%
      linkcolor=[rgb]{0.5,0.0,0.0},%
      citecolor=[rgb]{0,0.2,0.445},%
      filecolor=[rgb]{0,0,0.4},
      anchorcolor=[rgb]={0.0,0.1,0.2}%
   }
   
\usepackage{makecell}

\title{Constant-Hop Spanners for More Geometric Intersection Graphs, with Even Smaller Size}
   
\author{Timothy M. Chan}{Department of Computer Science, University of Illinois at Urbana-Champaign, USA}{tmc@illinois.edu}{https://orcid.org/0000-0002-8093-0675}{Work supported by NSF Grant CCF-2224271.}
\author{Zhengcheng Huang}{Department of Computer Science, University of Illinois at Urbana-Champaign, USA}{zh3@illinois.edu}{}{}

\titlerunning{Constant-Hop Spanners for More Geometric Intersection Graphs}
\authorrunning{T.\,M. Chan and Z. Huang}

\Copyright{Timothy M. Chan and Zhengcheng Huang}



\EventEditors{Erin W. Chambers and Joachim Gudmundsson}
\EventNoEds{2}
\EventLongTitle{39th International Symposium on Computational Geometry (SoCG 2023)}
\EventShortTitle{SoCG 2023}
\EventAcronym{SoCG}
\EventYear{2023}
\EventDate{June 12--15, 2023}
\EventLocation{Dallas, Texas, USA}
\EventLogo{}
\SeriesVolume{258}
\ArticleNo{24}

\newcommand{\R}{\mathbb{R}}
\newcommand{\Real}{\mathbb{R}}

\newcommand{\polylog}{\mathop{\rm polylog}}
\newcommand{\UU}{{\cal U}}
\newcommand{\IGNORE}[1]{}

\begin{document}

\maketitle

\begin{abstract}
In SoCG 2022, Conroy and T\'oth presented several constructions of sparse, low-hop spanners in geometric intersection graphs, including an $O(n\log n)$-size 3-hop spanner for $n$ disks (or fat convex objects) in the plane, and an $O(n\log^2 n)$-size 3-hop spanner for $n$ axis-aligned rectangles in the plane.  Their work left open two major questions: (i) can the size be made closer to linear by allowing larger constant stretch? and (ii) can near-linear size be achieved for more general classes of intersection graphs?

We address both questions simultaneously, by presenting new constructions of constant-hop spanners that have \emph{almost} linear size and that hold for a \emph{much larger} class of intersection graphs.  More precisely, we prove the existence of an $O(1)$-hop spanner for arbitrary \emph{string graphs} with $O(n\alpha_k(n))$ size for any constant~$k$, where $\alpha_k(n)$ denotes the $k$-th function in the inverse Ackermann hierarchy.  We similarly prove the existence of an $O(1)$-hop spanner for intersection graphs of $d$-dimensional fat objects with $O(n\alpha_k(n))$ size for any constant $k$ and $d$.

We also improve on some of Conroy and T\'oth's specific previous results, in either the number of hops or the size: we describe an $O(n\log n)$-size 2-hop spanner for disks (or more generally objects with linear union complexity) in the plane, and an $O(n\log n)$-size 3-hop spanner for axis-aligned rectangles in the plane.

Our proofs are all simple, using separator theorems, recursion, shifted quadtrees, and shallow cuttings.
\end{abstract}

\section{Introduction}

\emph{Spanners}---subgraphs of a given graph that preserve distances up to
some multiplicative factor---have numerous applications and have been studied extensively 
in both the graph algorithms 
and the computational geometry literature \cite{AhmedBSHJKS20,NarSmidBOOK}.
Traditionally, in computational geometry, the focus has been on Euclidean
spanners or metric spanners (i.e., spanners for a weighted complete
graph defined by $n$ points,
where the edge weights are Euclidean distances or distances under some metric).

Recently, spanners for geometric intersection graphs
have gained more attention.  A \emph{geometric intersection graph} is an
unweighted, undirected graph formed
by $n$ geometric objects, where the vertices are the objects, and we place an edge between two objects iff they intersect.
Such graphs are popularly studied in computational geometry
(e.g., see \cite{BringmannKKNP22,CabelloJ15,ChanS19,ChanS19a,GaoZ05}), and arise naturally in applications to wireless
communication.

Formally, in the unweighted setting, a \emph{$t$-hop spanner} of a graph $G$ is a subgraph $\widehat{G}$ of $G$, such that for each edge $uv\in E(G)$, there is a path of at most $t$ edges in $\widehat{G}$ from $u$ to $v$.  (It is sometimes just called a \emph{$t$-spanner}, but the term ``hop'' emphasizes that we are considering unweighted graph distances here.) The parameter $t$ is called the \emph{hop stretch}. For an arbitrary unweighted graph with $n$ vertices,
it is known~\cite{AlthoferDDJS93} that there exists a $t$-hop spanner with $O(n^{1+1/\lceil t/2\rceil})$ size (i.e., number of edges) for any constant integer $t\ge 3$; this bound is tight
assuming the Erd\H os girth conjecture~\cite{Erdos63}.  Our goal is to obtain better bounds
in the setting of geometric intersection graphs.

\subparagraph*{Previous results.}
Several papers studied hop spanners in the case of \emph{unit-disk graphs}, i.e.,
intersection graphs of unit disks in the plane:
Yan et al.~\cite{YanXD09} obtained 15-hop spanners with $O(n\log n)$ size.
Catusse et al.~\cite{CatusseCV10} obtained 5-hop spanners with $O(n)$ size
(with improvements on the hidden constant factor in the size bound
by Biniaz~\cite{Biniaz20} and Dumitrescu et al.~\cite{DumitrescuGT22}).
Dumitrescu et al.~\cite{DumitrescuGT22} also obtained 3-hop spanners with $O(n)$ size
and 2-hop spanners with $O(n\log n)$ size.
Finally, Conroy and Toth~\cite{conroy2022hop} obtained 2-hop spanners with $O(n)$ size.

Conroy and Toth~\cite{conroy2022hop} also initiated the study of hop spanners
for other families of geometric intersection graphs.  They obtained:

\begin{itemize}
\item 2-hop spanners for fat rectangles\footnote{All rectangles, squares, and hypercubes are axis-aligned throughout this paper.} (e.g., squares)
in the plane with $O(n\log n)$ size. (In fact, they proved a nearly matching lower bound
of $\Omega(n\log n/\log\log n)$ for squares, or for homothets of any fixed convex object in the plane.)
\item 3-hop spanners for fat convex objects (e.g., disks) in the plane with
$O(n\log n)$ size.
\item 3-hop spanners for arbitrary rectangles in the plane with
$O(n\log^2n)$ size.
\end{itemize}

\subparagraph*{Main questions.}
Conroy and T\'oth's work represented significant progress on hop spanners in
geometric intersection graphs, but it also raised a number of intriguing
questions:

\begin{enumerate}
\item Can the size of hop spanners be made closer to linear for the classes of
graphs they considered?  Their bounds for arbitrary disks, rectangles, etc.\ all have
extra logarithmic factors.  At the end of their paper, Conroy and T\'oth explicitly asked: ``is there a constant $t\in\mathbb{N}$ for which
every intersection graph of $n$ disks or rectangles admits a 
$t$-hop spanner with $O(n)$ edges?''
\item Ignoring logarithmic factors, can near-linear size hop spanners be
obtained for larger classes of geometric intersection graphs than the ones they
considered?  In particular, no $O(n\polylog n)$ size bounds were known for arbitrary line segments or
arbitrary triangles in the plane, or arbitrary balls in $\R^d$ for $d\ge 3$.
At the end of their paper, Conroy and T\'oth wrote: ``it would be interesting to see other classes of intersection graphs (e.g., for
strings or convex sets in $\R^2$, set systems with bounded VC-dimension
or semi-algebraic sets in $\R^d$) for which the general bound
of $O(n^{1+1/\lceil t/2\rceil})$ edges for $t$-hop spanners can be improved''.
\end{enumerate}

To appreciate the difficulty of these questions, it is worth mentioning
the connection to biclique cover size.  A \emph{biclique cover} of a graph $G$
refers to a collection of bicliques $A_1\times B_1,\ldots, A_s\times B_s$,
such that $E(G)=\bigcup_{i=1}^s (A_i\times B_i)$.  The \emph{size} of the cover
refers to $M=\sum_{i=1}^s (|A_i|+|B_i|)$.  Biclique covers are a standard
technique closely related to range searching, and have many applications in computational geometry (e.g., see \cite{AgarwalAAS94,Chan06,Chan23}).  Most classes of geometric intersection graphs
admit biclique covers with subquadratic size; in fact, for rectangles or axis-aligned
boxes, there are standard constructions of 
biclique covers with $O(n\polylog n)$ size (similar to the construction of range trees~\cite{AgaEriSURV,BerBOOK}).
Given a biclique cover of size $M$, it is easy to build a 3-hop spanner
of size $O(M)$, as noted by Conroy and T\'oth~\cite{conroy2022hop}, by just keeping
two stars per biclique.  In particular, Conroy and T\'oth's 3-hop, $O(n\log^2n)$-size spanners
for rectangles were obtained essentially by using range-tree-style divide-and-conquer.

However, biclique cover constructions typically require multiple logarithmic factor.  This makes the first question challenging for rectangles.
For non-axis-aligned objects, the biclique
cover size is even larger; for example, for line segments, the known upper bound is near $n^{4/3}$
(e.g., see~\cite{Chan23}).
Thus, this would not yield better bounds than for general graphs for hop stretch $t\ge 5$.  The situation is even worse for \emph{string graphs}, i.e.,
intersection graphs of curves (which could have large description complexity) in the plane.
New ideas are needed to address the second question.

\subparagraph*{Main new results.}
We make progress towards both of the above questions at once, by obtaining the following result:

\begin{itemize}
\item $O_k(1)$-hop spanners with  $O(n\alpha_k(n))$ size for arbitrary string graphs.
\end{itemize}

\noindent
Here, subscripts in the $O$ notation indicate variables that are assumed to be constant; the hidden constant factor may depend on such variables.  The function $\alpha_k(\cdot)$ denotes the $k$-th function in the inverse Ackermann hierarchy:
$\alpha_0(n)=n/2$, $\alpha_1(n)=\log n$, $\alpha_2(n)=\log^* n$ (the iterated logarithm), 
$\alpha_3(n)=\log^{**}n$ (the iterated iterated logarithm), etc.
Since these functions are extremely slow-growing as $k$ increases,
we thus get constant-hop spanners with \emph{almost} linear size.
Although inverse Ackermann has arisen in some past work on Euclidean spanners before
(namely, on the trade-off between size and hop-diameter~\cite{AryaDMSS95,LeMS22}), its appearance
here for hop spanners in geometric intersection graphs is still surprising.

String graphs include
intersection graphs of arbitrary regions enclosed by closed curves in the plane (e.g., see \cite[Lemma~4]{lee2017separators}).
Thus, our result is very general, encompassing arbitrary line segments and
triangles in $\R^2$ for the first time, and also including all the previous types of geometric
objects considered by Conroy and T\'oth, such as disks, rectangles, and fat convex objects in $\R^2$---our result shows that fatness is not needed in $\R^2$.

We obtain a similar result also for \emph{higher}-dimensional fat objects:

\begin{itemize}
\item $O_k(1)$-hop spanners with $O_{d,k}(n\alpha_k(n))$ size  for intersection graphs
of fat objects in $\R^d$.
\end{itemize}

\noindent
In particular, this includes the case of arbitrary balls in $\R^d$, for which 
there were no prior results for $d\ge 3$.  (In $d=2$ dimensions, compared to Conroy and T\'oth's previous result on fat objects, we work with a different definition of fatness that does not require convexity.)

\subparagraph*{More new results.}
The above new results improve previous size bounds for sufficiently large hop stretch, but do not necessarily improve Conroy and T\'oth's 
2-hop and 3-hop spanners.  We have additional new constructions that
directly improve some of their specific results.  Notably, we obtain:

\begin{itemize}
\item 2-hop spanners with $O(n\log n)$ size for objects with linear union complexity in the plane.
\end{itemize}

\noindent
Classes of objects with linear union complexity include
arbitrary disks, pseudodisks, and fat rectangles in $\R^2$.
Thus, our result strictly improves the hop stretch in Conroy and T\'oth's 3-hop, $O(n\log n)$-size spanners for the case of disks,
and also generalizes their 2-hop, $O(n\log n)$-size spanners for the case of fat rectangles. 
So, our result significantly enlarges the class of geometric intersection graphs that admit sparse 2-hop spanners.

In addition, we obtain:

\begin{itemize}
\item 3-hop spanners with $O(n\log n)$ size for rectangles in the plane.
\end{itemize}

\noindent
This is a logarithmic-factor improvement over Conroy and T\'oth's previous result.

A summary of  our results can be found in Table~\ref{tbl:contribution}.

\subparagraph*{Techniques.}
Our proofs use interesting techniques.
For string graphs, our approach (see Section~\ref{sec:string}) is based on
divide-and-conquer via \emph{graph separators}.
Separator theorems for string graphs have been developed in a series of papers~\cite{fox2010separator,matousek14,lee2017separators},
but sublinear-size separators exist only if the string graph is not too dense.
On the other hand, if the graph is dense, there exist high-degree vertices,
whose neighborhoods form large stars.  The key is to realize that each such star can be viewed 
as a single object, since a connected union of strings can be regarded
as a new string.  Besides recursion in the parts produced by the separator,
we use an extra recursive call to handle these stars.  (Luckily, the separator bound 
for string graphs is not influenced by how complicated the strings are.)
This double recursion eventually leads to inverse Ackermann complexity.  The overall construction is simple.

Curiously, even if we are only interested in the very special case of vertical/horizontal
line segments (or rectangles), it is still important to generalize to strings with
the above approach.  (In fact, we started this research more modestly with
the case of vertical/horizontal segments, using more traditional 
divide-and-conquer, but the above string-graph separator approach
wins out at the end.)

For fat objects in $\R^d$, the approach is similar, except that we use
\emph{shifted quadtrees} \cite{Bern93,Chan98,chan2003polynomial} and tree partitioning \cite{frederickson1997ambivalent} to do divide-and-conquer (see Section~\ref{sec:fat}).
The key is to view a union of fat objects containing a common point as 
a new fat object.  Again, we get a double recursion leading to inverse
Ackermann.

For our 2-hop spanners for objects with linear union complexity in $\R^2$,
we use logarithmically many layers of \emph{shallow cuttings} (see Section~\ref{sec:union}).
Shallow cuttings~\cite{matousek1992reporting}
have many applications, for example, to
static data structures for halfspace range searching~\cite{Chan00} and orthogonal  range searching~\cite{AfshaniT18},
dynamic geometric data structures~\cite{Chan20a}, levels in arrangements~\cite{Chan00}, incidences~\cite{ChanH23},  epsilon-nets~\cite{matousek1992reporting}, and
geometric set cover~\cite{chekuri2012set}.  Interestingly, our work adds
one more (unexpected) application to the list.  Given the shallow cutting lemma,
our proof is simple, this time, not even needing
recursion; in fact, it is simpler than 
Conroy and T\'oth's previous proofs for their 2-hop spanners for fat rectangles,
as well as their 3-hop spanners
for fat convex objects in $\R^2$.

Our 3-hop spanners for rectangles in $\R^2$ (see Section~\ref{sec:rect}) is perhaps the least exciting.
It is similar to Conroy and T\'oth's previous proof, using  straightforward
range-tree-style divide-and-conquer, but exploiting known spanners in 
one dimension (namely, points and intervals on the real line) as a base case.

\begin{table*}[t] 
\label{tab:results}
\centering
\begin{tabular}{|l|l|l|}
\hline
     & hop stretch & size
\\\hline
\multirow{2}{*}{String graphs}
    & 3 & $O(n\log^3 n)$
\\\cline{2-3}
    & $O_k(1)$ & $O(n\alpha_k(n))$
\\\hline
\multirow{2}{*}{Fat objects in $\R^d$}
    & 3 & $O(n\log n)$
\\\cline{2-3}
    & $O_k(1)$ & $O_{d,k}(n\alpha_k(n))$
\\\hline
    Objects with union complexity $\UU(\cdot)$ & 2 & $O(\UU(n)\log n)$
\\\hline
    Rectangles in $\R^2$ & 3 & $O(n\log n)$
\\\hline
\end{tabular}
\vspace{1ex}
\caption{Our results on $O(1)$-hop spanners for different classes of geometric intersection graphs.}
\label{tbl:contribution}
\end{table*}

\section{String Graphs}\label{sec:string}

Our spanner constructions for string graphs will use a separator theorem by Lee \cite{lee2017separators} (which was an improvement over previous versions by
Fox and Pach~\cite{fox2010separator} and
Matou\v sek~\cite{matousek1992reporting}).

\begin{lemma}[String-graph separator~\cite{lee2017separators}]
    \label{thm:string-graph-separator}
        For every string graph $G$ with $n$ vertices and $m$ edges, there
        exists a partition of $V(G)$ into subsets $V_1,V_2,X$ with
        $|V_1|,|V_2|\le 2n/3$, $|X|=O(\sqrt{m})$, 
        such that there are no edges between $V_1$ and $V_2$.
\end{lemma}

We first warm up by describing a 3-hop spanner with
$O(n\log^3n)$ size, and a 7-hop spanner with
$O(n\log\log n)$ size, before generalizing it to a
larger-hop spanner with inverse Ackermann complexity.

\subsection{3-Hop Spanner with $O(n\log^3n)$ Size}

\begin{center}
    {\bf String Graph Construction I}
\end{center}
\begin{enumerate}
    \item Repeatedly pick a vertex with degree larger than $\Delta$, for a parameter $\Delta$ to be chosen later, and remove the vertex along with its neighborhood, until there are no vertices with at most degree $\Delta$ in the remaining graph.  Let $G'$ be the remaining graph. A vertex and its neighborhood forms a star, and at most $n/\Delta$ such stars are removed. Add the edges of these stars ($O(n)$  in total) to the output spanner $\widehat{G}$.

    \item For each vertex $u\in V(G)$ and for each star removed in step 1 that contains a vertex adjacent to~$u$, add an edge between $u$ and an arbitrary such vertex in the star to $\widehat{G}$. At most $O(n\cdot n/\Delta)$ edges are added this way.

    \item Apply Lemma~\ref{thm:string-graph-separator} to $G'$ to obtain $V_1,V_2,X$.
    Since $G'$ has $O(\Delta n)$ edges, $|X|=O(\sqrt{\Delta n})$. Recursively construct a 3-hop spanner for the subgraph induced by $V_1\cup X$ and for the subgraph induced by $V_2\cup X$.  Add all their edges to $\widehat{G}$.
\end{enumerate}

\subparagraph*{Hop stretch.} For any edge $uv$, if both $u$ and $v$ belong to $G'$, then $u$ and $v$ are connected by 3 hops by induction.  Otherwise, one of its vertices, say, $v$, belongs to a star removed in step~1.  By step 2, the spanner $\widehat{G}$ connects $u$ to some vertex $v'$ in the same star as $v$. The star connects $v'$ and $v$ by 2 hops, so $u$ and $v$ are connected by 3 hops (see Figure \ref{fig:string-graph-i}).

\subparagraph*{Sparsity.} 
The size $S(n)$ of the spanner follows the recurrence
\begin{equation*}
    S(n)\ \leq\ \max_{n_1,n_2\le 2n/3:\ n_1+n_2\le n} (S(n_1 + O(\sqrt{\Delta n})) + S(n_2 + O(\sqrt{\Delta n})) + O(n\cdot n/\Delta)).
\end{equation*}
By setting $\Delta =n/\log^2 n$, the recurrence solves to $S(n)=O(n\log^3 n)$.

\begin{figure}
    \centering
    \includegraphics[page=2,scale=1.1]{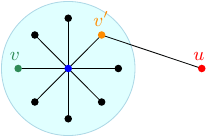}
    \caption{The light blue circle represents the star containing $v$. The blue dot at the center represents the center of the star.}
    \label{fig:string-graph-i}
\end{figure}

\begin{theorem}\label{thm:string-graph:3hop}
    Every string graph with $n$ vertices admits a $3$-hop spanner of size $O(n\log^3 n)$.
\end{theorem}

\subsection{7-Hop Spanner with $O(n\log\log n)$ Size}

To obtain hop spanners with still smaller size,
we need a generalized version of Theorem \ref{thm:string-graph-separator} that partitions 
into multiple subsets (analogous to ``$r$-divisions'' in planar graphs~\cite{frederickson1987fast}):

\newcommand{\VV}{\overline{V}}

\begin{lemma}[String-graph separator: multiple-subsets version]
    \label{lem:string-graph-r-division}
    Given parameters $r$ and $\Delta$ with $\Delta=o(r)$, for every string graph $G$ with $n$ vertices and maximum degree $\Delta$, there exist $O(n/r)$ subsets $\VV_i\subset V(G)$ of at most $r$ vertices each, such that
    $E(G)\subset \bigcup_i (\VV_i\times\VV_i)$, and the 
    \emph{boundary complexity}, defined as $\sum_i |\VV_i| - |V(G)|$,  
     is at most $O(\sqrt{\Delta}n/\sqrt{r})$.   
\end{lemma}
\begin{proof}
Apply Lemma \ref{thm:string-graph-separator} to obtain $V_1,V_2,X$, with $|X|=O(\sqrt{\Delta n})$, and
recursively generate subsets for the subgraph induced by $V_1\cup X$ and for
the subgraph induced by $V_2\cup X$.
When a subgraph has fewer than $r$ vertices, output its vertex set.

    Let $B(n)$ count 
    the boundary complexity of the subsets produced by the above division procedure on an $n$-vertex string graph. We have the recurrence
    \begin{equation*}
        B(n) \leq \left\{
        \begin{array}{ll}\displaystyle\max_{n_1,n_2\le 2n/3:\ n_1+n_2\le n}
        (B(n_1+O(\sqrt{\Delta n})) + B(n_2+O(\sqrt{\Delta n})) + O(\sqrt{\Delta n})) & \mbox{if $n\ge r$}\\
        0 & \mbox{if $n<r$}
        \end{array}\right.
    \end{equation*}
    The recurrence solves to $B(n)=O(\sqrt{\Delta} n/\sqrt{r})$.
\end{proof}

\begin{center}
    {\bf String Graph Construction II}
\end{center}
\begin{enumerate}
    \item Follow step 1 of Construction I.
    \item For each vertex $u\in V(G)$ that is adjacent to a vertex of at least one star, add an edge between $u$ and an arbitrary such vertex in such a star to $\widehat{G}$. At most $O(n)$ edges are added this way.
    \item For each pair of stars removed in step 1 such that there is a 2-hop path between them, add an arbitrary such 2-hop path to $\widehat{G}$.
    At most $O((n/\Delta)^2)$ edges are added this way.
    \item Apply Lemma~\ref{lem:string-graph-r-division} to $G'$ to obtain the subsets $\VV_i$.
    Recursively construct a 7-hop spanner for the subgraph induced by each $\VV_i$.  Add all their edges to $\widehat{G}$.
\end{enumerate}

\subparagraph*{Hop stretch.} For any edge $uv$, 
if both $u$ and $v$ belong to $G'$, then $u$ and $v$ are connected by 7 hops by induction.  Otherwise,
one of its vertices, say, $v$ belongs to a star removed in step~1.  By step 2, the spanner $\widehat{G}$ connects $u$ to some vertex $v'$ in a possibly different star.  By step~3, these two stars are connected by 2 hops in $\widehat{G}$.  Inside a star, 2 hops suffice.
Thus, $u$ and $v$ are connected by 7 hops (see Figure~\ref{fig:string-graph-ii}).

\begin{figure}
    \centering
    \includegraphics[page=3,scale=1]{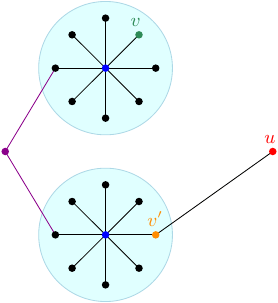}
    \caption{If $uv\in E(G)$, then $\widehat{G}$ connects $u$ and $v$ by 7 hops. The purple paths represent the arbitrary 2-hop path added in step~3.}
    \label{fig:string-graph-ii}
\end{figure}

\subparagraph*{Sparsity.} The size of the spanner follows the recurrence
\begin{equation*}
    S(n) \ \leq \max_{n_1,n_2,\ldots\le r:\ \sum_i 
 n_i\le n + O(\sqrt{\Delta}n/\sqrt{r})} \left(\sum_i S(n_i) + O(n+(n/\Delta)^2)\right).
\end{equation*}
By choosing $\Delta=\sqrt{n}$ and $r=n^{0.9}$, the recurrence solves to $S(n) = O(n\log\log n)$.

\begin{theorem}
    \label{thm:string-graph-basic}
    Every string graph with $n$ vertices admits a $7$-hop spanner of size $O(n\log\log n)$.
\end{theorem}

\subsection{$O_k(1)$-Hop Spanner with $O(n\alpha_k(n))$ Size}

Let $t_1=3$ and $t_k=5t_{k-1}+3$ for all $k>1$.
This implies $t_k=\frac{3}{4}(5^k-1)$. We next modify the preceding construction to obtain a $t_k$-hop spanner. The key is the following observation:

\begin{observation}\label{obs:string}
The union $C$ of strings that form a connected string graph may be viewed as a string.
\end{observation}
\begin{proof}
We can first eliminate cycles from $C$ by removing infinitesimally small arcs, and can then take an Euler traversal of the resulting tree to obtain a noncrossing path.  (The observation becomes even more obvious if one defines a string as any connected set in the plane, as some authors did~\cite{fox2010separator}.)
\end{proof}

\begin{center}
    {\bf String Graph Construction III}
\end{center}
\begin{enumerate}
    \item Follow step 1 of Construction I.
    \item Apply Lemma~\ref{lem:string-graph-r-division} to $G'$ to obtain the subsets $\VV_i$.  Let $B$ be the boundary vertices, i.e., vertices that are in at least two subsets $\VV_i$.
    Recursively construct a $t_k$-hop spanner for the subgraph induced by  $\VV_i\setminus B$ for each $\VV_i$.  Add all their edges to $\widehat{G}$.  Also, for each vertex in $B$, create a star of size 1 (i.e., a singleton) and remove it from $G'$.  The number of stars is now $O(n/\Delta + \sqrt{\Delta}n/\sqrt{r})$.
    \item For each vertex $u$ that is adjacent to a vertex of at least one star, add an  edge between $u$ and an arbitrary such vertex in such a star to $\widehat{G}$; we say that $u$ is \emph{assigned} to this star.  At most $O(n)$ edges are added this way.
    \item For each star $\sigma$, define its \emph{extended star} $S(\sigma)$ to be the set of all vertices that are in $\sigma$ or assigned to $\sigma$, and define
    the new object $U(\sigma)$ to be the union of all the strings in $S(\sigma)$.  Recursively construct a $t_{k-1}$-hop spanner $\widehat{H}$ for these new objects, which can be viewed as strings by Observation~\ref{obs:string}.  For each edge $U(\sigma)U(\sigma')$ in the spanner $\widehat{H}$, add an edge $ww'$ to $\widehat{G}$, where $w\in S(\sigma)$ and $w'\in S(\sigma')$ are intersecting strings chosen arbitrarily.
\end{enumerate}

\subparagraph*{Hop stretch.} 
For any edge $uv$, 
if both $u$ and $v$ belong to $G'$, then $u$ and $v$ are connected by $t_k$ hops by induction.
Otherwise, one of its vertices, say, $v$ belongs to a star $\sigma$ removed in step~1 or~2.  By step 3, the spanner $\widehat{G}$ connects $u$ to some vertex $v'$ in a possibly different star $\sigma'$.
Since $u$ is in $S(\sigma')$ and $v$ is in $\sigma$, the two objects $U(\sigma)$ and $U(\sigma')$ intersect and,  by induction, 
are connected by $t_{k-1}$ hops in the spanner $\widehat{H}$.  Inside an extended star, 4 hops suffice. Thus, $u$ and $v$ are connected by $5t_{k-1}+3$ hops in $\widehat{G}$ (see Figure \ref{fig:string-graph-iii}).

\subparagraph*{Sparsity.} The size of the $t_k$-hop spanner follows the recurrence
\begin{equation*}
    S_k(n)\ \leq \max_{n_1,n_2,\ldots\le r:\ \sum_i n_i \le n} \left( \sum_i S_k(n_i) + S_{k-1}(O(n/\Delta + \sqrt{\Delta}n/\sqrt{r})) + O(n)\right).
\end{equation*}
For the base case, we have $S_1(n)=O(n\log^3n)$ by Theorem~\ref{thm:string-graph:3hop}.
For $k=2$, by choosing $\Delta=\log^3n$ and $r=\Delta^3$,
the recurrence gives $S_2(n)=O(n\log^*n)$.
For $k>2$, we choose $\Delta=c_0\alpha_{k-1}(n)$
and $r=\Delta^3$ for a sufficiently large constant $c_0$.
It is straightforward to show by induction
that $S_k(n)\le c_0n\alpha_k(n)$.

\begin{figure}
    \centering
    \includegraphics[page=1,scale=1.1]{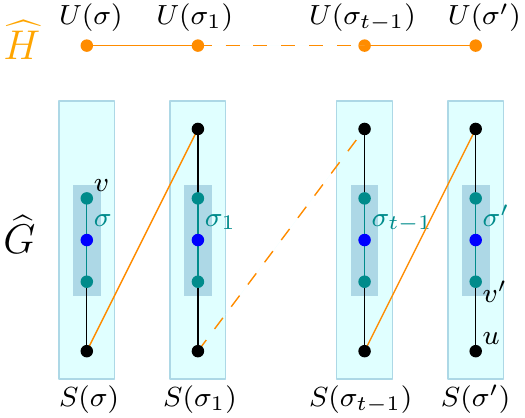}
    \caption{$v\in S(\sigma)$ and $v'\in S(\sigma')$. If $\widehat{H}$ connects $U(\sigma)$ and $U(\sigma')$ by $t$ hops, then $\widehat{G}$ connects $v$ and $v'$ by $5t+3$ hops. The smaller boxes represent stars in $\widehat{G}$; the larger boxes represent extended stars. Blue dots at the center of the boxes represent centers of the stars.}
    \label{fig:string-graph-iii}
\end{figure}

\begin{theorem}
    \label{thm:string-graph-ackerman}
     Every string graph with $n$ vertices admits a $\frac{3}{4}(5^k-1)$-hop spanner of size $O(n\alpha_k(n))$ for any $k\ge 2$.
\end{theorem}

\subparagraph*{Remarks.}
We have not attempted to optimize the hop stretch in the above theorem.  Since the constant factor in the above size bound does not depend on $k$, we can also choose $k=\alpha(n)$ and obtain an $O(5^{\alpha(n)})$-hop spanner with $O(n)$ size.

Although we have cited Lee's string-graph separator theorem~\cite{lee2017separators}, the weaker separator bound
by Fox and Pach~\cite{fox2010separator} is actually sufficient to prove Theorems~\ref{thm:string-graph-basic}--\ref{thm:string-graph-ackerman} (although for Theorem~\ref{thm:string-graph:3hop}, the bound would have more logarithmic factors).

\section{Fat Objects in $\R^d$}\label{sec:fat}

In this section, we turn our attention to the case of fat objects.  We will use the following definition of fatness~\cite{chan2003polynomial}.
Here, the \emph{side length} of an object refers to the side length of its smallest enclosing hypercube.

\begin{definition}
A collection of objects is \emph{$c$-fat} if for every hypercube $\gamma$ with side length $\ell$,
there exist $c$ points hitting all objects that intersect $\gamma$ and have
side length at least $\ell$.
\end{definition}

Our spanner construction will use quadtrees together with 
a known ``shifting lemma''~\cite{chan2003polynomial} (based
on an earlier work~\cite{Chan98}).

\begin{definition}
A \emph{quadtree cell} is a hypercube of the
form $[i_1/2^k,(i_1+1)/2^k)\times\cdots\times[i_d/2^j,(i_d+1)/2^k)$ for integers $i_1,\ldots,i_d,k$.

An object $u$ of side length $\ell$ is \emph{$C$-aligned} if it is contained in 
a quadtree cell with side length at most $C\ell$. 
\end{definition}

\begin{lemma}[Quadtree shifting~\cite{chan2003polynomial}]
    \label{lem:quadtree-shifting}
    Fix an odd number $d^*>d$. Let $\tau^{(j)}=(j/d^*,\ldots,j/d^*)\in\R^d$.
    For any object $u\subset [0,1)^d$, the shifted object $u+\tau^{(j)}$ is $(2d^*)$-aligned for all but at most $d$ indices $j\in\{0,\ldots,d^*-1\}$.
\end{lemma}

Choose $d^*=2d+1$.  By rescaling, we may assume that all input objects are in $[0,1)^d$. Then for any pair of objects $u$ and $v$, there exists at least one index $j\in\{0,\ldots,d^*-1\}$ such that $u+\tau^{(j)}$ and $v+\tau^{(j)}$ are both
$(2d^*)$-aligned.  For each $j$, it suffices to construct a hop spanner for the subset of all objects $u$ such that $u+\tau^{(j)}$ is $(2d^*)$-aligned; we can then output the union of these $d^*$ spanners.

Thus, from now on, we may assume that
all given objects are $(2d^*)$-aligned.

We warm up by describing a 3-hop spanner with $O(n\log n)$ size, before describing a larger-hop spanner with inverse Ackermann complexity.

\subsection{3-Hop Spanner of $O(n\log n)$ Size}

Our 3-hop spanner will use the following lemma, which follows directly by
taking a tree centroid in the quadtree:

\begin{lemma}[Quadtree centroid~\cite{AryaMNSW98,Chan98}]\label{lem:quadtree:centroid}
For any set of $n$ points in $\R^d$,
there exists a quadtree cell such that 
the number of points inside and the number of points outside
are both at most $\frac{2^d}{2^d+1}n$.
\end{lemma}

\begin{center}
    {\bf Fat Object Construction I}
\end{center}
\begin{enumerate}
    \item Apply Lemma~\ref{lem:quadtree:centroid} to the leftmost points of the objects to obtain a quadtree cell $\gamma$.
Recursively construct a 3-hop spanner for the objects completely inside $\gamma$ 
and for the objects completely outside $\gamma$.
Add all their edges to $\widehat{G}$.
    \item
    Let $P_\gamma$ be a set of points hitting all objects that intersect $\partial\gamma$ and have side length at least $\ell_\gamma/(2d^*)$, where $\ell_\gamma$ denotes the side length of $\gamma$.  A hitting set of size $|P_\gamma|=O_d(c)$ exists
    by definition of $c$-fatness (since $\partial\gamma$ can be covered by $(2d^*)^d$ hypercubes of side length $\ell_\gamma/(2d^*)$).
    For each point $p\in P_\gamma$, build a star $S(p)$ connecting all objects hit by $p$, with the center chosen arbitrarily.  Add the edges of these stars  to $\widehat{G}$.  At most $O_d(cn)$ edges are added this way.
    \item 
    For each object $u$ and for each star $S(p)$ that contains an object intersecting $u$,
    add an edge between $u$ and an arbitrary such object in $S(p)$ to $\widehat{G}$.  At most $O(n)$ edges are added this way.   
\end{enumerate}

\subparagraph*{Hop stretch.} For any edge $uv$, if the objects $u$ and $v$ are both inside $\gamma$ or both outside~$\gamma$, then $u$ and $v$ are connected by 3 hops by induction.  Otherwise, one of the objects, say, $v$, intersects $\partial\gamma$ (see Figure \ref{fig:fat-object-i}).  Observe that $v$ has side length
at least $\ell_\gamma/(2d^*)$, since $v$ is $(2d^*)$-aligned.
Thus, $v$ belongs to a star $S(p)$ from step~2. By step 3, the spanner $\widehat{G}$ connects $u$ to some object $v'$ in the same star $S(p)$.  Since $v'$ and $v$ are connected by 2 hops in $\widehat{G}$, $u$ and $v$ are connected by 3 hops.

\subparagraph*{Sparsity.} 
The size $S(n)$ of the spanner follows the recurrence
\begin{equation*}
    S(n)\ \leq\ \max_{n_1,n_2\le 2^dn/(2^d+1):\ n_1+n_2\le n} (S(n_1) + S(n_2) + O_d(cn)).
\end{equation*}
The recurrence solves to $S(n)=O_d(cn\log n)$.

\begin{theorem}
    \label{thm:fat-object-graph-d-dimension-basic}
    The intersection graph of $n$ fat objects in $\Real^d$ admits a 3-hop spanner of size $O(n\log n)$.
\end{theorem}

\begin{figure}
    \centering
    \includegraphics[page=1,scale=1.2]{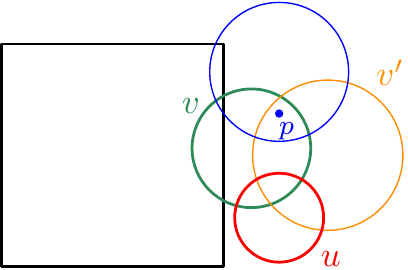}
    \caption{$uv'\in\widehat{G}$, and $v'$ is connected to $v$ via the center of star $S(p)$ (the blue disk).}
    \label{fig:fat-object-i}
\end{figure}

\subsection{$O_k(1)$-Hop Spanner with $O_{d,k}(n\alpha_k(n))$  Size}


To obtain hop spanners of still smaller size, we need a generalized version of Lemma~\ref{lem:quadtree:centroid} that partitions into multiple subsets:

\begin{definition}
A \emph{generalized quadtree cell} $\gamma$ refers to either a quadtree cell or the difference of an outer quadtree cell $\gamma^+$ and an inner quadtree cell $\gamma^-$ (in the former case, we let
$\gamma^+=\gamma$ and $\gamma^-=\emptyset$).
\end{definition}

\begin{lemma}[Quadtree partitioning]\label{lem:quadtree:partition}
Given parameter $r$,
for any set of $n$ points in $[0,2)^d$, there exists a partition of $[0,2)^d$ into $O(n/r)$ generalized quadtree cells, each containing at most $r$ points.
\end{lemma}
\begin{proof}
This follows directly by applying the
tree partitioning scheme by Frederickson \cite{frederickson1997ambivalent} to the quadtree,
or alternatively by applying Lemma~\ref{lem:quadtree:centroid} recursively (stopping when cells have at most $r$ points each, and with further splitting to ensure each generalized cell has at most one inner quadtree cell---e.g., see \cite{AryaMNSW98}).
\end{proof}

\newcommand{\ggamma}{\widehat{\gamma}}

\begin{observation}
    \label{obs:fat-object-union}
    Given a collection of $c$-fat $C$-aligned  objects in $\R^d$,
     a union of a subset of objects all hit by a common point can be viewed as a $(4^dc)$-fat $C$-aligned object.
\end{observation}
\begin{proof}
Consider a hypercube $\gamma$ with side length~$\ell$.
Expand $\gamma$ into a hypercube $\ggamma$ with side length $2\ell$, keeping the same center.
There exists a set $P$ of $4^dc$ points hitting all objects that intersect $\ggamma$ and have side length at least $\ell/2$.

Now take a subset $S$ of objects containing a common point $p_0$.  Let $U$ be the union of the objects in $S$.
Suppose that $U$ intersects $\gamma$ and has side length at least $\ell$.
\begin{itemize}
\item Case 1: $p_0\in\ggamma$.
Some object $u\in S$ has side length at least $\ell/2$.
Since $u$ contains $p_0$ and thus intersects $\ggamma$,
we know that $u$ is hit by $P$, and so $U$ is hit by $P$.
\item Case 2: $p_0\not\in\ggamma$.
Some object $u\in S$ intersects $\gamma$, and $u$ must have side length at least $\ell/2$.  Thus, $u$ is hit by $P$, and so $U$ is hit by $P$.
\end{itemize}
This proves $(4^dc)$-fatness of $U$.  The $C$-alignedness of $U$ follows from the $C$-alignedness of the individual objects in $S$.
\end{proof}

Let $t_1=3$
and $t_k=3t_{k-1}+3$.
This implies
$t_k=\frac{11}{9}3^k-\frac{2}{3}$.  We now describe a construction of a $t_k$-hop spanner.

\begin{center}
    {\bf Fat Object Construction II}
\end{center}
\begin{enumerate}
    \item Apply Lemma~\ref{lem:quadtree:partition} to the leftmost points of the objects to obtain a set $\Gamma$ of $O(n/r)$ generalized quadtree cells.
    For each $\gamma\in\Gamma$, construct a $t_k$-hop spanner recursively for the objects completely inside $\gamma$.

    \item     For each $\gamma\in\Gamma$, let $P_{\gamma^+}$ be a set of $O_d(c)$ points hitting all objects that intersect $\partial\gamma^+$ and have side length at least $\ell_{\gamma^+}/(2d^*)$, where $\ell_{\gamma^+}$ denotes the side length of $\gamma^+$.  Similarly, let $P_{\gamma^-}$ be a set of $O_d(c)$ points hitting all objects that intersect $\partial\gamma^-$ and have side length at least $\ell_{\gamma^-}/(2d^*)$, where $\ell_{\gamma^-}$ denotes the side length of $\gamma^-$.
    Let $P_\gamma=P_{\gamma^+}\cup P_{\gamma^-}$.
    For each object $u$ completely inside $\gamma$ and each point $p\in P_\gamma$,
    add an edge between $u$ and an arbitrary object that is hit by $p$ and intersects $u$ (if exists) to $\widehat{G}$.  At most $O_d(cn)$ edges are added this way.

    \item
    Let $P=\bigcup_{\gamma\in\Gamma} P_\gamma$.
    \emph{Assign} each object $u$ that is hit by $P$ to an arbitrary $p\in P$ that hits~$u$.  For each $p\in P$, build a star $S(p)$ connecting all objects assigned to $p$, with the center chosen arbitrarily.  Add the edges of these stars to $\widehat{G}$.  At most $O(n)$ edges are added this way.

    \item For each $p\in P$, define the new object $U(p)$ to be the union of the objects in $S(p)$. Recursively construct a $t_{k-1}$-hop spanner $\widehat{H}$ for these new objects, which are $(4^dc)$-fat and $(2d^*)$-aligned by Observation \ref{obs:fat-object-union}.
    For each edge $U(p)U(p')$ in the spanner $\widehat{H}$, add an edge $ww'$ to $\widehat{G}$, where $w\in S(p)$ and $w'\in S(p')$ are intersecting objects chosen arbitrarily.
\end{enumerate}

\subparagraph*{Hop stretch.}
For any edge $uv$, if both $u$ and $v$ are completely inside a generalized quadtree cell in $\Gamma$, then
$u$ and $v$ are connected by $t_k$ hops by induction.  

Otherwise, consider the case where neither $u$ nor $v$ are completely inside a generalized quadtree cell in $\Gamma$.  Then $u$ intersects $\partial\gamma$ and
$v$ intersects $\partial\gamma'$ for some $\gamma,\gamma'\in\Gamma$.
Observe that $u$ has side length at least $\ell_{\gamma^+}/(2d^*)$ (resp.\ $\ell_{\gamma^-}/(2d^*)$) if $u$ intersects $\partial\gamma^+$
(resp.\ $\partial\gamma^-$), because $u$ is $(2d^*)$-aligned.  So, $u$ is hit by $P_\gamma$.  Similarly, $v$ is hit by $P_{\gamma'}$.  Thus, $u$ and $v$ belong to two stars $S(p)$ and $S(q)$ from step~3 for some $p,q\in P$.  The two objects $U(p)$ and $U(q)$ intersect and, by induction, are connected by $t_{k-1}$ hops in the spanner $\widehat{H}$.  Inside a star, 2 hops suffice.
Thus, $u$ and $v$ are connected by $3t_{k-1}+2$ hops in $\widehat{G}$ (see Figure \ref{fig:fat-object-ii}).

Lastly, consider the case when exactly one of the objects, say, $u$, is completely inside a generalized quadtree cell $\gamma$ in $\Gamma$, and the other object $v$ intersects $\partial\gamma$. Then, $v$ is hit by some point $q\in P_\gamma$. By step~2, the spanner $\widehat{G}$ connects $u$ to some object $v'$ that is hit by the same point $q$. Then $v$ and $v'$ belong to two stars $S(p)$ and $S(p')$ from step~3 for some $p,p'\in P$. The two objects $U(p)$ and $U(p')$ intersect. By the same argument in the previous case, $v$ and $v'$ are connected by $3t_{k-1}+2$ hops in $\widehat{G}$, and so $u$ and $v$ are connected by $3t_{k-1}+3$ hops.

\begin{figure}
    \centering
    \includegraphics[page=2,scale=1.1]{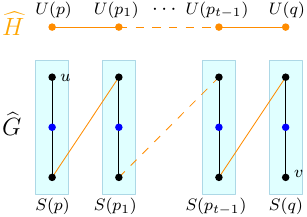}
    \caption{Assume that $u\in S(p)$ and $v\in S(q)$ for some $p,q\in P$. If $\widehat{H}$ connects $U(p)$ and $U(q)$ by $t$ hops, then $\widehat{G}$ connects $u$ and $v$ by $3t+2$ hops. Stars in $\widehat{G}$ are shown as blue boxes. The blue dots represent the centers of these stars.}
    \label{fig:fat-object-ii}
\end{figure}

\subparagraph*{Sparsity.} 
The size of the $t_k$-hop spanner for $c$-fat $(2d^*)$-aligned objects follows the recurrence
\begin{equation*}
    S_{k,c}(n)\ \leq \max_{n_1,n_2,\ldots\le r:\ \sum_i n_i \le n} \left( \sum_i S_{k,c}(n_i)  + S_{k-1,4^dc}(O_d(cn/r)) + O_d(cn)\right).
\end{equation*}
For the base case, we have $S_{1,c}(n)=O_d(cn\log n)$
by Theorem~\ref{thm:fat-object-graph-d-dimension-basic}.
For $k>1$, we choose $r=\alpha_{k-1}(n)$.  It is straightforward to verify by induction that $S_{k,c}(n)=O_{d,k,c}(n\alpha_k(n))$.

\begin{theorem}
    \label{thm:fat-object-graph-ackerman}
    The intersection graph of $n$ $c$-fat objects in $\Real^d$ admits a $(\frac{11}{9}3^k-\frac{2}{3})$-hop spanner of size $O_{d,k,c}(n\alpha_k(n))$ for any $k\ge 1$.
\end{theorem}

\subparagraph*{Remarks.}
Because the fatness parameter $c$ grows as a function of $k$ during recursion, the constant factor in the above size bound depends on $k$.  

Like in the previous section, it is also possible to obtain an intermediate result, namely, a 6-hop spanner with $O(n\log\log n)$ size.

\section{Objects with (Near) Linear Union Complexity in $\R^2$}\label{sec:union}

In this section, we describe a different approach
to construct hop spanners, using the shallow cutting lemma 
introduced by Matou\v sek~\cite{matousek1992reporting}.  The variant below 
can be found in~\cite{chekuri2012set}. 

\begin{lemma}[Shallow cutting]
    \label{thm:shallow-cutting-disks}
    Consider a family of \emph{well-behaved}\footnote{
    See ~\cite{chekuri2012set} for a precise definition.
    Most families of objects in $\R^2$, such as disks, pseudodisks, etc.\ are well-behaved. 
    }
    objects in $\R^2$,
    such that the union of any $n$ objects has complexity at most $\UU(n)$, assuming that $\UU(n)/n$ is nondecreasing.
    Given a set of $n$ objects in this family and parameters $r$ and $k$, there exists a collection of $O((rk/n+1)^2\UU(n/k))$ cells, such that (i)~each cell intersects the boundaries of at most $n/r$ objects, and (ii)~the cells cover all points of depth at most $k$.  Here, the \emph{depth} of a point $p$ is the number of objects that contain $p$.
\end{lemma}

\begin{center}
    {\bf Construction via Shallow Cuttings}
\end{center}
\begin{enumerate}
    \item For each $i=1,\ldots,\log n$,
    apply Lemma~\ref{thm:shallow-cutting-disks} with $k=2^i$ and $r=n/2^{i-2}$ to obtain a collection $\Xi_i$ of $O(\UU(n/2^i))$ cells.  We may assume that each cell $\Xi_i$ contains at least one point of depth at most $2^i$ (otherwise, the cell may be removed).
    \item For each $i$ and for each $\xi\in\Xi_i$ such that there exists an object $s(\xi)$ that contains $\xi$ completely, build a star centered at $s(\xi)$ connecting all objects that intersect $\xi$. Add the edges of these stars to $\widehat{G}$. Since there are at most $2^i$ objects that contain $\xi$ and $2^{i-2}$ objects whose boundaries intersect $\xi$, the number of edges added is $O(\sum_{i=1}^{\log n} \UU(n/2^i)\cdot 2^i) = O(\UU(n)\log n)$.
\end{enumerate}

\subparagraph*{Hop stretch.} 
For any edge $uv$, pick an arbitrary point $p$ in the intersection of $u$ and $v$, and let $k$ be the depth of $p$. Let $i$ be the number such that $2^{i-1}\leq k<2^i$. Let $\xi$ be the cell in $\Xi_i$ that contains $p$.  At least $2^{i-1}$ objects contain $p$, but at most $2^{i-2}$ objects have boundaries intersecting 
$\xi$.  Thus, there must
exist an object $s(\xi)$ that completely contains $\xi$.
Then, $\widehat{G}$ contains the edges $s(\xi)u$ and $s(\xi)v$, and so $u$ and $v$ are connected by 2 hops (see Figure~\ref{fig:pseudo-disks}).

\begin{figure}
    \centering
    \includegraphics[scale=1.2]{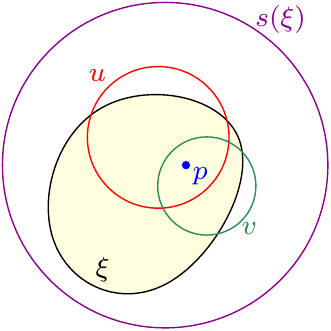}
    \caption{Objects $u$ and $v$ intersect, $p$ is an arbitrary point in $u\cap v$, and $\xi\in\Xi_i$ is the cell containing $p$. Our counting argument shows that there must exist an object $s(\xi)$ completely containing $\xi$, and thus $\widehat{G}$ must contain the path $u\rightarrow s(\xi)\rightarrow v$.}
    \label{fig:pseudo-disks}
\end{figure}

\begin{theorem}
    \label{thm:pseudo-disk-graph}
    Consider a family of well-behaved objects in $\R^2$,
    such that the union of any $n$ objects has complexity at most $\UU(n)$, assuming that $\UU(n)/n$ is nondecreasing.
    The intersection graph of $n$ objects in this family admits a 2-hop spanner of size $O(\UU(n)\log n)$.
\end{theorem}

For example, it is known that $\UU(n)=O(n)$ for disks and pseudodisks in the plane~\cite{KedemLPS86}, and $\UU(n)=O(n\log^*n)$ for fat triangles in the plane~\cite{AronovBES14}.  So, we obtain 2-hop spanners of size $O(n\log n)$ for disks and pseudodisks, and 
size $O(n\log n\log^*n)$ for fat triangles.

\subparagraph*{Remarks.}  
It is possible to reduce the size bound to $O(\UU(n)\log\log n)$ with 5 hops (by using fewer shallow cuttings, with $k=2^{(1+\delta)^i}$, for $i=1,\ldots,O(\log\log n)$), but this approach does not appear to yield further improvement for larger hop stretch.

\section{Axis-Aligned Rectangles in $\R^2$}\label{sec:rect}

In this section, we describe a 3-hop spanner for the case where the input objects consist of horizontal line segments $H$ and vertical line segments $V$ in the plane.  Spanners for the more general case of axis-aligned rectangles will then follow.

We first consider the special case where all vertical segments are lines. This problem is 1-dimensional in the sense that the $y$-coordinates of the segments are irrelevant. Borrowing Conroy and T\'oth's technique \cite{conroy2022hop} for 1D interval graphs, we divide the $x$-axis into disjoint intervals $\mathcal{I}=\{I_1,\ldots,I_\ell\}$ as follows (see Figure \ref{fig:seg-line}):

\begin{enumerate}
    \item $I_0=\{x_0\}$ is the interval containing only the $x$-coordinate of the leftmost endpoint among all horizontal segments.
    \item For integers $k\geq 1$, $I_k=(x_{k-1},x_k]$, where $x_{k-1}$ is the right boundary of $I_{k-1}$, and $x_k$ is the largest number for which there exists a line segment $h_k=[x'_k,x_k]\in H$ such that $x'_k\leq x_{k-1}$. We say that $h_k$ is the {\em covering segment} of $I_k$.
\end{enumerate}

\begin{figure}
    \centering
    \includegraphics[scale=0.75]{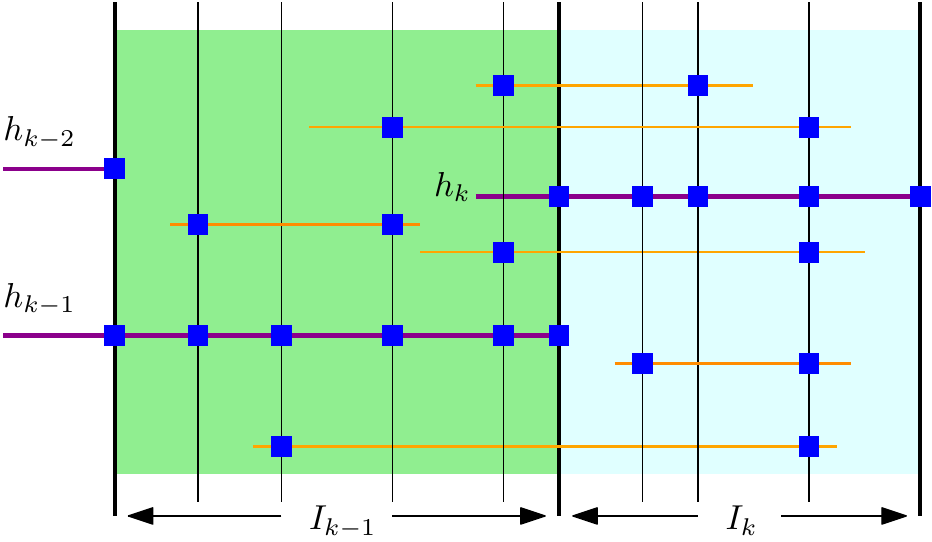}
    \caption{Illustration of the disjoint intervals $I_{k-1}$ and $I_{k}$, and the covering segments $h_{k-2},h_{k-1},h_k$ (drawn in purple). 
    The blue squares indicate which intersections are kept.}
    \label{fig:seg-line}
\end{figure}

\begin{lemma}
    \label{lem:seg-line}
    The intersection graph $G$ of $n$ horizontal segments and vertical lines admits a 3-hop spanner $\widehat{G}$ with $O(n)$ edges.
\end{lemma}
\begin{proof}
    For each interval $I=(x_L,x_R]$, let $h_I$ be the covering segment. Keep all intersections that involve $h_I$ in the slab $I\times\Real$. Finally, for every segment $h\in H$, keep the intersections with the leftmost and the rightmost vertical line that intersects $h$, denoted $v_L(h)$ and $v_R(h)$ respectively. Let $\widehat{G}$ be the subgraph that includes an edge for each intersection we keep.
    
    \subparagraph{Hop stretch.} Consider $h\in H$ and $v\in V$ that intersect. Either $v_L(h)$ or $v_R(h)$ is in the same interval $I\in\mathcal{I}$ as $v$; say it is $v_R(h)$. Both $v$ and $v_R(h)$ intersect the covering segment $h_I$. Thus, $\widehat{G}$ contains the 3-hop path $h\rightarrow v_R(h)\rightarrow h_I\rightarrow v$.
    
    \subparagraph{Sparsity.}
    For each covering segment $h_I$, we have only kept intersections in the interval $I$, so we have kept $O(n)$ intersections over all intervals $I$. For each $h\in H$ that is not a covering segment, we have kept only two intersections involving $h$.
\end{proof}

Using Lemma \ref{lem:seg-line}, the standard binary divide-and-conquer along the $y$-axis gives us a 3-hop spanner with $O(n\log n)$ edges for the case of horizontal and vertical line segments. Given a horizontal slab $\sigma$, we construct the 3-hop spanner as follows:

\begin{enumerate}
    \item Construct a 3-hop spanner according to Lemma \ref{lem:seg-line} to handle the intersections between horizontal segments and \emph{long} vertical segments, i.e., vertical segments that cross the entire slab $\sigma$.  Then remove the long vertical segments.
    \item Divide $\sigma$ into two horizontal subslabs, each containing half the number of horizontal segments.
    For each of the two subslabs,
    construct a 3-hop spanner recursively.
\end{enumerate}

Each segment, whether horizontal or vertical, appears in $O(\log n)$ of the recursive calls. Therefore, the total number of edges in the spanner is bounded by $O(n\log n)$. Thus, we have proved the following:

\begin{lemma}
    The intersection graph of $n$ horizontal/vertical segments admits a 3-hop spanner of size $O(n\log n)$. 
\end{lemma}

We can extend the results for axis-aligned line segments to axis-aligned rectangles by replacing each rectangle with four line segments, each being one side of the rectangle. We build a spanner 
for these line segments. If two rectangles intersect, then either their sides intersect, or one rectangle contains the other. The first case reduces to segment intersection. For the case of containment, Conroy and T\'oth~\cite{conroy2022hop} have shown that using $O(n\log n)$ edges, there is a 2-hop spanner for the subgraph that includes only ``corner intersections'', i.e., intersections where one rectangle contains a corner of the other rectangle.

\begin{theorem}
    The intersection graph of $n$ axis-aligned rectangles in $\Real^2$ admits a 3-hop spanner of size $O(n\log n)$.
\end{theorem}

\section{Open Questions}

Although we have obtained almost linear size bounds for
hop spanners in string graphs and fat-object intersection graphs, a remaining question is whether these upper bounds could be further improved to linear, or whether an inverse-Ackermann-type lower bound could be proved.  

Another question is whether near-linear bounds are possible
for other intersection graphs not addressed here, e.g., for simplices in $\R^3$.  Here, one might want to start more modestly with any upper bound better than for general graphs.



\bibliographystyle{plainurl}
\bibliography{reference}

\end{document}